\documentclass[12pt,onecolumn]{IEEEtran}
\usepackage[T1]{fontenc}
\usepackage{amssymb}
\usepackage{amsmath}
\usepackage{amsthm}
\usepackage{amsfonts}
\usepackage{amsthm}
\usepackage{amsmath}
\usepackage{amsfonts}
\usepackage{amssymb}
\usepackage{bbm}
\usepackage{graphicx}
\usepackage{hyperref}
\usepackage{algorithm,algpseudocode}
\usepackage{mathtools}
\usepackage{dsfont}
\usepackage{verbatim}
\usepackage{subcaption}
\usepackage{float}
\usepackage{eucal}
\usepackage{balance}
\usepackage{bm}
\usepackage{amssymb}
\usepackage{amsmath}
\usepackage{amsthm}
\usepackage{amsfonts}
\usepackage{hyperref}
\usepackage{bbold}
\usepackage{mathtools}
\usepackage{dsfont}
\usepackage{verbatim}
\usepackage{subcaption}

\newtheorem{theorem}{Theorem}

\title{Reconfigurable Intelligent Surfaces Aided Communication: Capacity and Performance Analysis Over Rician Fading Channel}

\author{
\IEEEauthorblockN{Chandradeep Singh, Chia Hsiang Lin}
\vspace*{-0.65cm}
}

\begin{document}
\maketitle
\begin{abstract}
In this work, we consider a single input single output (SISO) system for Reconfigurable Intelligent Surface (RIS) assisted mmWave communication. We consider Rician channel models over user node to RIS and RIS to Access Point (AP). We obtain closed form expressions for capacity with channel state information (CSI) and without CSI at the transmitter. Newly derived capacity expressions are closed form expressions in a very compact form. We also simplified the closed form expressions for average symbol error probability. We also characterize the impacts of key parameters Rician factor $K$ and number of elements on IRS on ergodic capacity with CSI and without CSI at the transmitter.
\end{abstract}
\begin{IEEEkeywords}
Power Control, Water Filling, Ergodic Capacity, Reconfigurable Intelligent Surface (RIS). 
\end{IEEEkeywords}

\section{Introduction}
\IEEEPARstart{T}{ransmission} at high frequencies is required to provide high data rate in next generation wireless networks. Transmission over mmWave is highly susceptible to path loss. To compensate for this path loss, highly directed beams are used for line of sight (LoS) communication using highly directed antenna arrays. Still these beams at high frequencies are attenuated or blocked by the objects in the LoS path between two communicating nodes. RIS is usually deployed to provide a virtual LoS path between two communicating nodes and increases the channel quality between the two nodes. RIS is the key technique that provides the small reconfigurable wireless environment for next generation wireless networks \cite{Wu2020}. RIS has energy efficient low cost reflecting elements which are controlled by the RIS controller.
 
Authors describe the key features of RIS technology, RIS consists of nearly passive elements, low cost reflecting elements with reconfigurable parameters to enhance the spectrum and energy efficiency of next generation mobile networks \cite{Basar2019}. RIS technology has better performance than existing wireless technologies such as multiple input and multiple output (MIMO) systems, multi antenna amplify and forward relaying systems with a small number of antennas \cite{Wu2019,Boulogeorgos2020}. Most of the literature on RIS technology focuses on active and passive beamforming \cite{Wu_2020,Huang2019,Hu2020,Zhou2020}. Upper bound on channel capacity and asymptotic expressions for outage probability over Rician fading models are derived for RIS assisted communication \cite{Tao2020}. Non-central chi-square distribution based approximations are used to estimate the channel distribution of the RIS and upper bounds on bit error rate are derived for indoor and outdoor communication \cite{Yildirim2021}. In \cite{Jia2020}, authors optimize phase shift for instantaneous CSI to maximize average rate and maximize ergodic rate by optimizing the phase shift for statistical CSI in RIS aided communication. In \cite{Guo2020}, authors derive expression for outage probability for maximal ratio transmission (MRT) and optimize the phase shift to minimize the outage. Characterisation of the sum of $\alpha - \mu$ distributed random variates has been studied in \cite{Kong2021}. Low cost adaptive phase shift design with LoS CSI in RIS aided communication has been proposed in \cite{Han2019}. In \cite{Kudathanthirige2020}, outage probability and achievable rate for single input and single output (SISO) systems have been derived for RIS aided communication. In \cite{Yang2020}, accurate approximations of channel distribution for Rayleigh fading and performance matrices are derived for RIS consisting of any number of elements. Recently, outage probability, average symbol error probability and channel capacity over the Rician fading model have been derived for more accurate closed form approximations \cite{Salhab2021}. 

While the aforementioned works study performance analysis, the power control for known CSI with Rician fading in RIS aided communication has not been covered in the literature. We use Laguerre series method given in \cite{Primak2004} to derive channel distribution then derive ergodic capacity with CSI at transmitter. We use a water filling algorithm with CSI for optimal power control as described in \cite{Goldsmith1997}. Our capacity expressions are more simplified than \cite{Salhab2021}.

The rest of the article is organized as follows: System model is discussed
in Section~\ref{sec:sys_model}. Capacity and performance analysis is described in Section~\ref{Perfom_analys}. 
Numerical evaluation of the closed form expressions is provided in Section~\ref{sec:ne}.  
Finally, we conclude in Section~\ref{sec:con}.

\section{System Model}
\label{sec:sys_model}

As shown in fig.~\ref{fig:sys_model} , we consider a single antenna at transmitter and single antenna at receiver (SISO system). We consider $h_l$ as channel gain from AP to reflecting element $l$ at RIS and $g_l$ as channel gain from element $l$ at RIS to user node. RIS consists of total $N$ reflecting elements. Phase shift induced by reflecting element $l$ is $\phi_l$. Channel vector for AP-RIS link is given by $\bm{h}=\left[h_1,\ldots,h_N\right]^T$. Similarly channel vector for RIS-user node link is given by $\bm{g}=\left[g_1,\ldots,g_N\right]^T$. Phase shift matrix $\bm{\phi} \triangleq \mbox{diag}\left(\phi_1,\ldots,\phi_N\right)\in \mathbb{C}^{N{\times}N}$. $P_t$ is transmitted power. Let $x$ is transmitted signal then received signal at user node is 
\begin{equation}
y =\sqrt{P_t}\bm{g}^T\bm{\phi}\bm{h}x+n, \label{eq:channelmodel}
\end{equation}
where $n$ is the additive white Gaussian noise (AWGN) $n\in\mathcal{CN}(0,N_0)$. Channel coefficient $h_l=\alpha_le^{\theta_l}$, where $\alpha_l$ is the amplitude and $\theta_l$ is the phase shift. Similarly, channel coefficient $g_l=\beta_le^{\theta'_l}$, where $\beta_l$ is the amplitude and $\theta'_l$ is the phase shift. Furthermore, we assume that the envelopes of first hop channels $\alpha_l$, where $l\in\{1,\ldots,N\}$ are independent identically distributed (i.i.d.) Rician random variable with shape parameter $K_1$ and scale parameter $\Omega_1$. Similarly,  envelopes of second hop channels $\beta_l$, where $l\in\{1,\ldots,N\}$ are i.i.d. Rician random variables with shape parameter $K_2$ and scale parameter $\Omega_2$. We consider that the $\bm{h}$ and $\bm{g}$ are independent. Let $v_i^2$ denotes power in the LoS component and $2\sigma_i^2$ denotes the power in the non line of sight (NLoS) component. Shape parameter $K_i=\frac{v_i^2}{2\sigma_i^2}$ and scale parameter $\Omega_i=v_i^2+2\sigma_i^2$.
\begin{figure}[h]
\vspace{1em}
\begin{centering}
\includegraphics[scale=0.8]{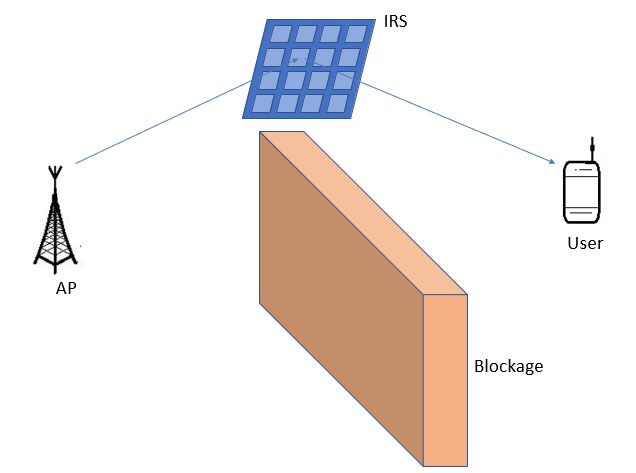}
\vspace{1.5em}
\caption{Communication aided by an RIS}
\label{fig:sys_model}
\end{centering}
\end{figure}

When transmitter transmits at average power $P_{av}$ then according to \cite{Basar2019}, maximum end-to-end SNR can be obtained as 
\begin{equation}\label{eq:snr}
\gamma =\overline{\gamma}\left(\sum_{l=1}^{N}\alpha_l\beta_l\right)^2,\\ 
\mbox{where }\overline{\gamma}=\frac{P_{av}}{N_0}. \nonumber
\end{equation}
Let $\xi=\sum_{l=1}^{N}\xi_l$ where $\xi_l=\alpha_l\beta_l$. Probability density function of $\xi_l$ can be obtained as  \cite[eq.~(24)]{Donoughue2012}
\begin{equation}\label{eq:pdf}
f_{\xi_l}\left(y\right) =\sum_{j=1}^{\infty}\sum_{i=1}^{\infty}\frac{y^{i+j+1}\left(4K_2^iK_1^j(\Omega_1\Omega_2)^{\frac{1}{2}(i+j+2)}\right)K_{j-i}(2y\sqrt{\Omega_1\Omega_2})}{(i!)^2(j!)^2\exp(K_1K_2)},\\ 
\end{equation}
where $K_{\nu}(\cdot)$ is the modified Bessel function of second kind and order $\nu$. Mean and variance of cascaded Rician random variable $\xi_l$ are evaluated as
\begin{align}\label{eq:mean}
\mathbbm{E}\left(\xi_l\right) &= \mathbbm{E}\left(\alpha_l\right)\mathbbm{E}\left(\beta_l\right) \mbox{ where } \nonumber \\
\mathbbm{E}\left(\alpha_l\right) &= \frac{1}{2}\sqrt{\frac{\Omega_1\pi}{K_1+1}}L_{1/2}(-K_1) \mbox{ and } \mathbbm{E}\left(\beta_l\right) = \frac{1}{2}\sqrt{\frac{\Omega_2\pi}{K_2+1}}L_{1/2}(-K_2). \nonumber \\
\end{align}
Where $L_{1/2}(\cdot)$ is Laguerre polynomial.  $L_{1/2}(x)=e^{x/2}\left[(1-x)I_0\left(\frac{-x}{2}\right)-xI_1\left(\frac{-x}{2}\right)\right]$,
where $I_{\nu}(\cdot)$ is the modified Bessel function of first kind and order $\nu$. 
\begin{align}
\mathbbm{E}\left(\xi_l\right)&=\frac{\pi e^{-\frac{(K_1+K_2)}{2}}}{4}\sqrt{\frac{\Omega_1\Omega_2}{(K_1+1)(K_2+1)}}\left[(K_1+1)I_0\left(\frac{K_1}{2}\right)+K_1I_1\left(\frac{K_1}{2}\right)\right] \nonumber \\
&\times \left[(K_2+1)I_0\left(\frac{K_2}{2}\right)+K_2I_1\left(\frac{K_2}{2}\right)\right],
\end{align}
\begin{align} \label{eq:var}
\mathbbm{E}\left(\xi_l^2\right)&=\mathbbm{E}\left(\alpha_l^2\right)\mathbbm{E}\left(\beta_l^2\right)=\Omega_1\Omega_2 \mbox{ and}\nonumber \\
\mathrm{Var}\left(\xi_l\right) &=\mathbbm{E}\left(\xi_l^2\right)-\left[\mathbbm{E}\left(\xi_l\right)\right]^2=\Omega_1\Omega_2-\left[\mathbbm{E}\left(\xi_l\right)\right]^2.
\end{align}

\begin{theorem}
PDF and cumulative density function of end-to-end SNR are given by 
\begin{equation}\label{eq:pdf_snr}
f_{\gamma}\left(\gamma\right) \simeq \frac{\gamma^{\frac{a-1}{2}}\exp\left(-\frac{\sqrt{\gamma}}{b\sqrt{\overline{\gamma}}}\right)}{2b^{a+1}\Gamma\left(a+1\right)\overline{\gamma}^{\frac{a+1}{2}}} \mbox{ and }
\end{equation}
\begin{equation}\label{eq:pdf_snr}
F_{\gamma}\left(\gamma\right) \simeq \frac{\gamma\left(a+1,\frac{\sqrt{\gamma}}{b\sqrt{\overline{\gamma}}}\right)}{\Gamma\left(a+1\right)} \mbox{ respectively}.
\end{equation}
Where $\gamma(\cdot,\cdot)$ is lower incomplete Gamma function given in \cite[eq.~8.350.1]{Gradshteyn2000}, 
$a=\frac{(\mathbbm{E}(\xi))^2}{\mathrm{Var}(\xi)}-1$ and $b=\frac{\mathrm{Var}(\xi)}{\mathbbm{E}(\xi)}$.
\end{theorem}
\begin{proof}
As $\xi$ is the sum of i.i.d random variables therefore the central limit theorem can estimate the PDF of $\xi$ using first term of Laguerre expansion given in \cite[Sec. 2.2.2]{Primak2004} as 
\begin{equation}\label{eq:pdf_si}
f_{\xi}\left(y\right) \simeq \frac{y^a}{b^{a+1}\Gamma\left(a+1\right)}\exp\left(-\frac{y}{b}\right).
\end{equation}
$\mathbbm{E}(\xi)=N\mathbbm{E}(\xi_l)$ and $\mathrm{Var}(\xi)=N \mathrm{Var}(\xi_l)$. Use relation $\gamma=\overline{\gamma}\xi^2$.
\end{proof}

\section{Performance and Capacity Analysis}\label{Perfom_analys}
\subsection{Average Symbol Error Probability}
For any modulation scheme, the performance is expressed in the form of average symbol error probability (ASEP). Expression for ASEP using CDF is given in \cite{McKay2007} as 
\begin{equation}\label{eq:err_prob}
ASEP=\frac{p\sqrt{q}}{2\sqrt{\pi}}\int_{0}^{\infty}\frac{\exp(-q\gamma)}{\sqrt{\gamma}}F_{\gamma}\left(\gamma\right)d\gamma,
\end{equation}
where $p$ and $q$ are modulation specific parameters.
\begin{theorem}
For RIS aided communication ASEP is given by
\begin{equation}\label{eq:err_probclos}
ASEP=\frac{2^{(a-1)}p}{\pi\Gamma(a+1)} G_{3,4}^{2,3} \left(\frac{1}{4q\overline{\gamma}b^2}\bigg|_{\frac{a+1}{2},\frac{a+2}{2},0,\frac{1}{2}}^{\frac{1}{2},\frac{1}{2},1}\right),
\end{equation}
where $G_{p,q}^{m,n}(\cdot)$ is Meijer's G function as given in \cite[Sec. 9.3]{Gradshteyn2000}. 
\end{theorem}
\begin{proof}
Eq. \eqref{eq:err_prob} can be written as 
\begin{equation}\label{eq:err_prob1}
ASEP=\frac{p\sqrt{q}}{2\sqrt{\pi}\Gamma\left(a+1\right)}\int_{0}^{\infty}\frac{\exp(-q\gamma)}{\sqrt{\gamma}}\gamma\left(a+1,\frac{\sqrt{\gamma}}{b\sqrt{\overline{\gamma}}}\right)d\gamma.
\end{equation}
As given in \cite[Sec. 8.4.16]{Prudnikov1998}, the lower incomplete Gamma function can be expressed as $\gamma(v,x)=G_{1,2}^{1,1}\left( x\big|_{v,0}^{1}\right)$. Therefore,
\begin{equation}\label{eq:err_prob2}
ASEP=\frac{p\sqrt{q}}{2\sqrt{\pi}\Gamma\left(a+1\right)}\int_{0}^{\infty}\frac{\exp(-q\gamma)}{\sqrt{\gamma}}G_{1,2}^{1,1}\left( \frac{\sqrt{\gamma}}{b\sqrt{\overline{\gamma}}}\bigg|_{a+1,0}^{1}\right)d\gamma.
\end{equation}
Using identity given in \cite[Sec. 2.24.3.1]{Prudnikov1998}  as below
\begin{eqnarray} \label{eq:identity}
\begin{split}
& \int_{0}^{\infty}x^{\alpha-1}e^{-\sigma x}G_{p,q}^{m,n}\left( \omega x^{l/k}\bigg|_{b_1,\ldots,b_m,\ldots,b_q}^{a_1,\ldots,a_n,\ldots,a_p}\right)dx=\frac{k^{\mu}l^{\alpha-1}\sigma^{-\alpha}}{(2\pi)^{\frac{l-1}{2}+c^*(k-1)}} \\
& G_{kp+l,kq}^{km,kn+l}\left( \frac{\omega^k l^l}{\sigma^l k^{k(q-p)}}\bigg|_{\frac{b_1}{k},\ldots, \frac{k+b_1-1}{k},\ldots,\frac{b_m}{k},\ldots, \frac{k+b_m-1}{k},\ldots,\frac{b_q}{k},\ldots, \frac{k+b_q-1}{k}}^{\frac{1-\alpha}{l},\ldots,\frac{l-\alpha}{l},\frac{a_1}{k},\ldots, \frac{k+a_1-1}{k},\ldots,\frac{a_n}{k},\ldots, \frac{k+a_n-1}{k},\ldots,\frac{a_p}{k},\ldots, \frac{k+a_p-1}{k}}\right),\\
&\mbox{ where } \mu=\sum_{j=1}^q  b_j-\sum_{j=1}^pa_j+\frac{p-q}{2}+1 \mbox{ and } c^*=m+n-\frac{p+q}{2}.
\end{split}
\end{eqnarray}
Finally we get,
\begin{equation}
ASEP=\frac{2^{(a-1)}p}{\pi\Gamma(a+1)} G_{3,4}^{2,3} \left(\frac{1}{4q\overline{\gamma}b^2}\bigg|_{\frac{a+1}{2},\frac{a+2}{2},0,\frac{1}{2}}^{\frac{1}{2},\frac{1}{2},1}\right). \nonumber
\end{equation}
\end{proof}
Derived expression for ASEP is very compact compared to \cite{Salhab2021} and easy to implement as Meijer's G function is available in MATLAB.
\subsection{Capacity Analysis}
When AP transmits at average power then capacity can be expresses in terms of PDF of $\gamma$ as 
\begin{equation}\label{eq:erg_cap}
\langle \hat{C} \rangle=\frac{1}{\ln2}\int_{0}^{\infty}\ln(1+\gamma)f_{\gamma}\left(\gamma\right)d\gamma.
\end{equation}
\begin{theorem}\label{thm:no_csi}
For RIS aided communication ergodic capacity without CSI is given by
\begin{equation}\label{eq:nocsi_cap}
\langle \hat{C} \rangle=\frac{2^{a}}{\ln2\sqrt{\pi}\Gamma(a+1)} G_{4,2}^{1,4} \left(4\overline{\gamma}b^2\bigg|_{1,0}^{\frac{-a}{2},\frac{-a+1}{2},1,1}\right).
\end{equation} 
\end{theorem}
\begin{proof}
As given in \cite[Sec. 8.4.6]{Prudnikov1998}, $\ln(1+x)$ can be expressed as $\ln(1+x)=G_{2,2}^{1,2}\left( x\big|_{1,0}^{1,1}\right)$. Therefore eq.~\eqref{eq:erg_cap} can be written as
\begin{equation}
\langle \hat{C} \rangle=\frac{1}{2\ln2(b^{a+1})\Gamma\left(a+1\right)\overline{\gamma}^{\frac{a+1}{2}}}\int_{0}^{\infty}\gamma^{\frac{a-1}{2}}\exp\left(-\frac{\sqrt{\gamma}}{b\sqrt{\overline{\gamma}}}\right)G_{2,2}^{1,2}\left( \gamma\big|_{1,0}^{1,1}\right)d\gamma.  \nonumber
\end{equation}
Changing the variable of integration as $\gamma=t^2$, we get
\begin{equation}
\langle \hat{C} \rangle=\frac{1}{\ln2(b^{a+1})\Gamma\left(a+1\right)\overline{\gamma}^{\frac{a+1}{2}}}\int_{0}^{\infty}t^a\exp\left(-\frac{t}{b\sqrt{\overline{\gamma}}}\right)G_{2,2}^{1,2}\left( t^2\big|_{1,0}^{1,1}\right)dt.  \nonumber
\end{equation}

Using the identity given in eq.~\eqref{eq:identity}, we get the desired expression.
\end{proof}
For given CSI, optimal power for transmission is given by water filling \cite{Goldsmith1997} as
\begin{equation}
  \frac{P_t(\gamma)}{P_{av}}= \begin{cases}
  \frac{1}{\gamma_{0}}-\frac{1}{\gamma}  & \mbox{if } \,\, \gamma>\gamma_{0} \\
  0,  & \mbox{otherwise}   ,
  \end{cases} \label{eq:update}\\
\end{equation}
where $\gamma_0$ is cutoff water level. $\gamma_0$ satisfies average power constraint $\int_{\gamma_0}^{\infty}\left(\frac{1}{\gamma_{0}}-\frac{1}{\gamma}\right)f_{\gamma}\left(\gamma\right)d\gamma=1$.
\begin{theorem}\label{thm:csi}
For RIS aided communication ergodic capacity using optimal power control for known CSI is given by
\begin{equation}\label{eq:csi_cap}
\langle C \rangle=\frac{2^{a}}{\ln2\sqrt{\pi}\Gamma(a+1)} G_{4,2}^{0,4} \left(\frac{4b^2}{\gamma_0}\bigg|_{0,0}^{\frac{-a}{2},\frac{-a+1}{2},1,1}\right),
\end{equation} 
\end{theorem}
\begin{proof}
As given in \cite{Goldsmith1997}, capacity using water filling algorithm is given by $\int_{\gamma_0}^{\infty}\ln\left(\frac{\gamma}{\gamma_0}\right)f_{\gamma}\left(\gamma\right)d\gamma$. Following \cite[Sec. 8.4.6]{Prudnikov1998}, $\ln(x)H(x-1)$ can be expressed as $\ln(x)H(x-1)=G_{2,2}^{0,2}\left( x\big|_{0,0}^{1,1}\right)$, where $H(x-1)$ is the unit step function. Following the proof of theorem~\ref{thm:no_csi}, we can obtain the similar expression.
\end{proof}
To the best of our knowledge, expressions \eqref{eq:err_probclos}, \eqref{eq:nocsi_cap}, \eqref{eq:csi_cap} are derived in compact form for the first time.

\section{Numerical Evaluations}
\label{sec:ne}
This section discusses the simulation results to show the performance of power control in RIS aided communication. In Fig. \ref{fig:asep}, we reproduce the results of \cite{Salhab2021} using a new closed form expression. We fix $K_1=K_2=1$, $N=5$ and vary $\Omega_1$ and $\Omega_2$ by changing $\sigma_1^2$ and $\sigma_2^2$. This figure shows that with increase in $\Omega$, coding gain increases and improves the performance of the system.

In our simulations for Fig. \ref{fig:element_2}, we consider $\sigma_1^2=\sigma_2^2=1/2$ and $N=2$. We evaluate ergodic capacity using power control for known CSI and capacity by transmitting at average power when CSI is not available. Simulation results show that when RIS consists of a very small number of elements then the water filling algorithm improves the capacity in a low SNR regime. Fig. \ref{fig:element_2} also shows that if the value of $K_1=K_2=K$ is less, it means channel is more random therefore power control helps in increasing the data rate but for higher values of $K$, transmitting at average power performs as good as water filling algorithm. This shows that using RIS for LoS communication does not need power control because the channel is less random in LoS communication.

\begin{figure}[h]
\vspace{1em}
\begin{centering}
\includegraphics[scale=1]{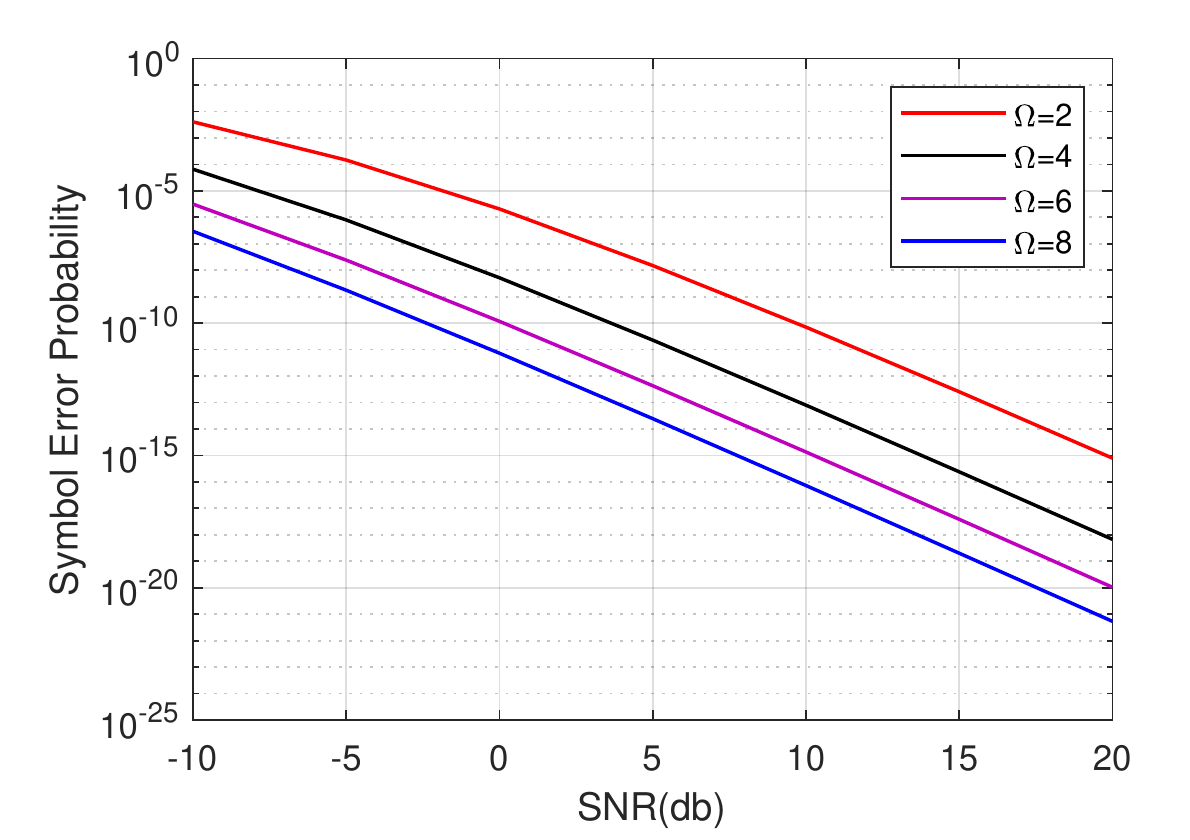}
\vspace{1.5em}
\caption{ASEP Vs SNR for different values of $K_1=K_2=1$ with $N=5$.}
\label{fig:asep}
\end{centering}
\end{figure}

\begin{figure}[h]
\vspace{1em}
\begin{centering}
\includegraphics[scale=1]{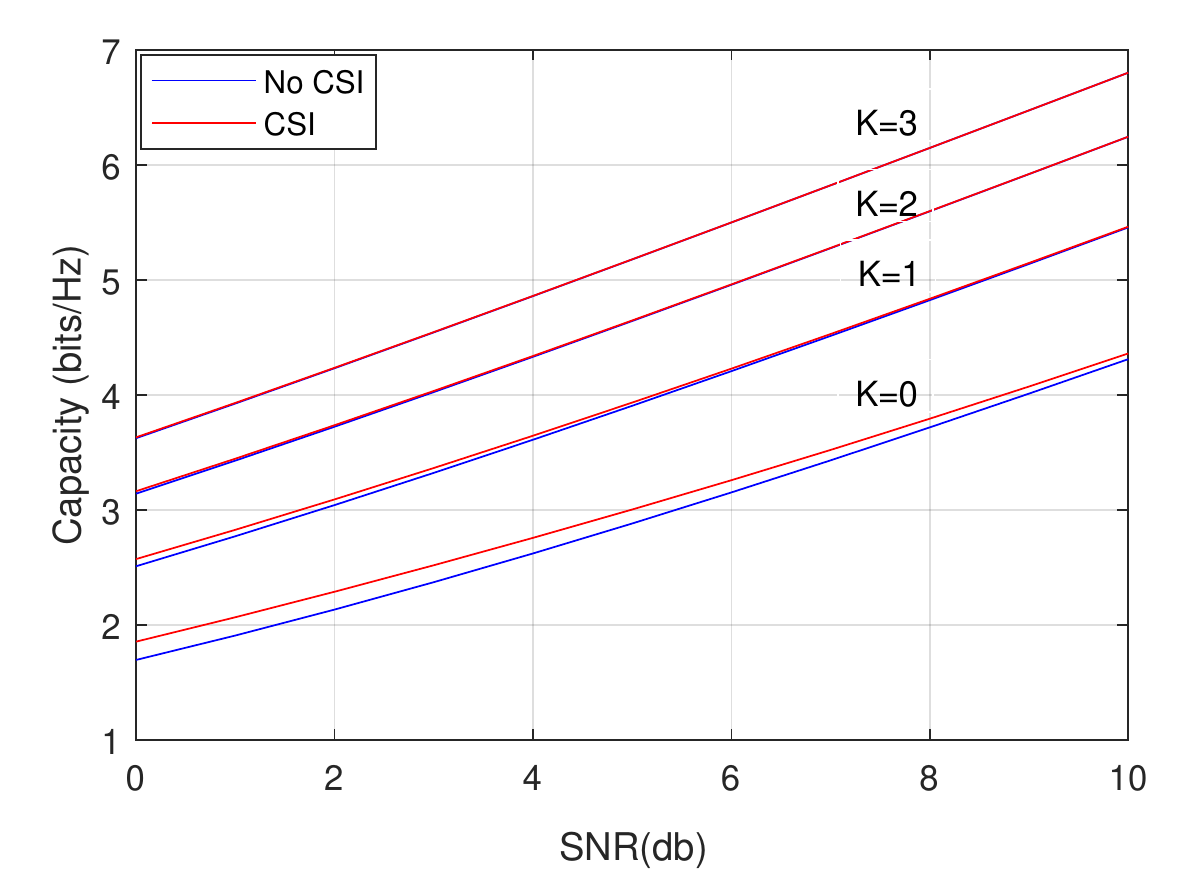}
\vspace{1.5em}
\caption{Capacity Vs SNR for different values of $K_1=K_2=K$ with $N=2$.}
\label{fig:element_2}
\end{centering}
\end{figure}
\indent\indent In Fig. \ref{fig:element_5}, we plot capacity vs SNR for different values of $K$ with $N=5$. Simulations results shows that, even for low values of Rician factor $K$ when channel is more random, increasing the elements on RIS compensate for randomness and therefore transmission at average power performs as good as the water filling algorithm. This shows that RIS with high number of reflecting elements does not need power control to improve capacity. Transmitter can transmit at average power if the number of elements is very large and it also reduces the overhead required for channel estimation.

\begin{figure}[h]
\vspace{1em}
\begin{centering}
\includegraphics[scale=1]{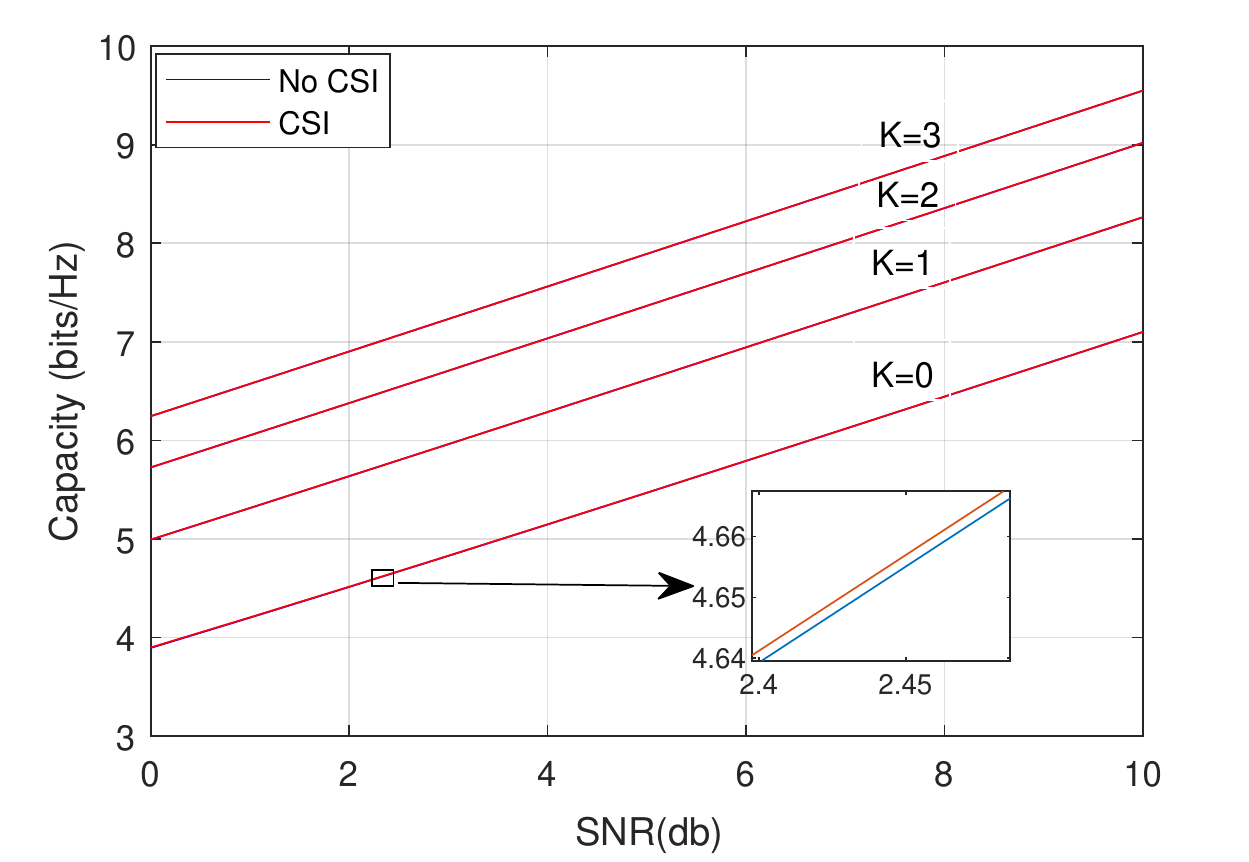}
\vspace{1.5em}
\caption{Capacity Vs SNR for different values of $K_1=K_2=K$ with $N=5$.}
\label{fig:element_5}
\end{centering}
\end{figure}

\section{Conclusion}
\label{sec:con}
In this work, we have investigated that LoS communication aided by an RIS generally transmits over a less random Rician fading channel. A transmitter does not need CSI to adapt power if RIS consists of a large number of elements in LoS communication. CSI is needed at RIS for passive beamforming. Adaptive power control with CSI at the transmitter can improve the performance of a system in a low SNR regime in a more random environment with a low value of Rician factor $K$.

\section{Acknowledgements}
\label{sec:ack}
Author also thanks the support from "NCKU 90 and Beyond" initiative at National Cheng Kung University, Taiwan.

\balance

\end{document}